\documentclass[journal]{IEEEtran}

\usepackage[noadjust]{cite}
\usepackage{graphicx}
\usepackage[cmex10]{amsmath}
\usepackage{amssymb}
\usepackage{cases}

\newtheorem{theorem}{Theorem}[section]
\newtheorem{lemma}[theorem]{Lemma}
\newtheorem{proposition}[theorem]{Proposition}
\newtheorem{corollary}[theorem]{Corollary}
\newtheorem{definition}[theorem]{Definition}
\newtheorem{example}[theorem]{Example}
\newtheorem{remark}[theorem]{Remark}

\newcommand{\findent}{\hspace{1em}}
\newlength{\nullrellength}
\newcommand{\nullrel}{\:\hspace{\nullrellength}\:}
\newcommand{\breakop}[1]{\hspace{-0.2222em}{}#1}

\newcommand{\relvar}[2]{\buildrel {#2} \over {#1}}
\newcommand{\eqvar}[1]{\relvar{=}{\mathrm{#1}}}

\newcommand{\eqdef}{\eqvar{\scriptscriptstyle\triangle}}

\newcommand{\nnintegers}{\mathbb{N}_0}

\newcommand{\complex}{\mathbb{C}}
\newcommand{\unitsof}[1]{{#1^{*}}}
\newcommand{\nzcomplex}{\unitsof{\complex}}
\newcommand{\field}[1]{\mathbb{F}_{#1}}
\newcommand{\nzfield}[1]{\unitsof{\field{#1}}}
\newcommand{\symgroup}[1]{\mathrm{S}_{#1}}

\newcommand{\seq}[1]{\mathbf{#1}} 
\newcommand{\mat}[1]{\mathbf{#1}} 
\newcommand{\matentry}[1]{\mathrm{#1}} 
\newcommand{\orth}[1]{#1^\perp} 
\newcommand{\weight}{\mathrm{w}} 
\newcommand{\sow}{\mathrm{w}_2} 
\newcommand{\sop}{\mathcal{S}} 
\newcommand{\sope}{S} 
\newcommand{\coef}{\mathrm{coef}}
\newcommand{\tr}{\mathrm{Tr}}
\newcommand{\repmap}{f^{\mathrm{rep}}}
\newcommand{\repcode}{C^{\mathrm{rep}}}
\newcommand{\chkmap}{f^{\mathrm{chk}}}
\newcommand{\chkcode}{C^{\mathrm{chk}}}
\newcommand{\ldmapA}{F^{\mathrm{ld1}}}
\newcommand{\ldcodeA}{\mathcal{C}^{\mathrm{ld1}}}
\newcommand{\ldmapB}{F^{\mathrm{ld2}}}
\newcommand{\ldcodeB}{\mathcal{C}^{\mathrm{ld2}}}
\newcommand{\rlccode}{\mathcal{C}^{\mathrm{rlc}}}
\newcommand{\rlccodeB}{\mathcal{C}^{\mathrm{rlc2}}}
\newcommand{\rmat}[1]{\tilde{\mathbf{#1}}} 
\newcommand{\matset}[1]{\mathcal{#1}} 

\begin{document}

\title{Second-Order Weight Distributions}
\author{Shengtian~Yang,~\IEEEmembership{Member,~IEEE}
\thanks{S. Yang is self-employed at Zhengyuan Xiaoqu 10-2-101, Hangzhou 310011, China (email: yangst@codlab.net).}
}

\maketitle

\markboth{Accepted for Publication in IEEE Transactions on Information Theory (Version: \today)}{}


\maketitle

\begin{abstract}
A fundamental property of codes, the second-order weight distribution, is proposed to solve the problems such as computing second moments of weight distributions of linear code ensembles. A series of results, parallel to those for weight distributions, is established for second-order weight distributions. In particular, an analogue of MacWilliams identities is proved. The second-order weight distributions of regular LDPC code ensembles are then computed. As easy consequences, the second moments of weight distributions of regular LDPC code ensembles are obtained. Furthermore, the application of second-order weight distributions in random coding approach is discussed. The second-order weight distributions of the ensembles generated by a so-called $2$-good random generator or parity-check matrix are computed, where a $2$-good random matrix is a kind of generalization of the uniformly distributed random matrix over a finite filed and is very useful for solving problems that involve pairwise or triple-wise properties of sequences. It is shown that the $2$-good property is reflected in the second-order weight distribution, which thus plays a fundamental role in some well-known problems in coding theory and combinatorics. An example of linear intersecting codes is finally provided to illustrate this fact.
\end{abstract}


\IEEEpeerreviewmaketitle

\section{Introduction}\label{sec:Introduction}

\IEEEPARstart{T}{he} weight distribution is an important property of a code. Knowledge about the weight distribution of a linear code is very useful for estimating decoding error probabilities \cite{Gallager196300} and covering radius \cite{Blinovsky199700}, etc.

Let $\field{q}$ be the finite field of order $q = p^r$, where $p$ is prime and $r \ge 1$. Throughout this paper, all codes considered will be over $\field{q}$. For a codeword $\seq{c} \in \field{q}^n$, its (Hamming) weight $\weight(\seq{c})$ is the number of nonzero symbols in $\seq{c}$. For a code $C \subseteq \field{q}^n$, the weight distribution $A_i(C)$ is the number of codewords of weight $i$ in $C$. The polynomial $W_C(x) \eqdef \sum_{i=0}^n A_i(C)x^i$ is called the weight enumerator of $C$.

In general, it is hard to compute the weight distribution of a specific code. On the other hand, for some ensembles of codes, we can compute their average weight distributions. For example, we have known the average weight distribution of various low-density parity-check (LDPC) code ensembles \cite{Gallager196300, Litsyn200204, Litsyn200312, Burshtein200406, Di200611, Richardson200800, Flanagan200911, Bennatan200403, Como2008, Yang201010}, and even the second moment of weight distributions of some specific LDPC code ensembles, e.g., \cite{Rathi200509, Barak200509, Rathi200609, Barak200711}. While the average weight distribution provides by Markov's inequality an upper bound of the weight distribution of codes in the ensemble with high probability, the degree to which individual codes have a weight distribution close to this average depends on second- or higher-order statistics of the random weight distribution. Typically, given a random code $\mathcal{C}$ (a measurable mapping from some abstract probability space to the power set of $\field{q}^n$), one computes the second moment $E[(A_i(\mathcal{C}))^2]$ of the weight distribution, or the second moment $E[W_{\mathcal{C}}(x)W_{\mathcal{C}}(y)]$ of the weight enumerator, since
\[
E[(A_i(\mathcal{C}))^2] = \coef(E[W_{\mathcal{C}}(x)W_{\mathcal{C}}(y)], x^iy^i)
\]
where $\coef(p(x,y), x^iy^j)$ denotes the coefficient of $x^iy^j$ in the polynomial $p(x,y)$. Applying Chebyshev's inequality with the second moment then gives a confidence interval of the weight distribution of an individual code with respect to any given probability. It is not the only case that one needs to compute the second moment like $E[W_{\mathcal{C}}(x)W_{\mathcal{C}}(y)]$. When estimating the variance of undetected error probability of an error detection scheme (see e.g. \cite[p.~99]{Moon200500} and \cite{Wadayama201005}) which is expressed in terms of a random weight enumerator $W_{\mathcal{C}}(x)$, one also needs to compute the second moment $E[(W_{\mathcal{C}}(x))^2]$ or $E[A_i(\mathcal{C})A_j(\mathcal{C})]$. Then one question arises: how can we compute these second moments? For this question, there has been some work for some specific code ensembles, e.g., the second moment of the weight distribution of a binary regular LDPC code ensemble \cite{Rathi200509, Barak200509, Rathi200609}, the covariance of the weight distribution of a linear code ensemble characterized by the so-called Bernoulli parity-check matrix ensemble \cite{Wadayama201005}, and the second moment of the weight distribution of a random linear code generated by a uniform random generator matrix \cite{Blinovsky200912}. However, no systematic approach has ever been established to facilitate such kinds of computation.

To establish a systematic approach, we need a fundamental property of linear codes, which not only yields the second moment of weight distribution but also supports easy computation for various combinations of linear codes. Unfortunately, the distribution $A_i(\mathcal{C})^2$ or $A_i(\mathcal{C})A_j(\mathcal{C})$ are not qualified for this position. When $\mathcal{C}$ is not random, it is clear that $E[A_i(\mathcal{C})A_j(\mathcal{C})]$ provides no more information than does $E[A_i(\mathcal{C})]$. On the other hand, for a general random $\mathcal{C}$, the information contained in $E[A_i(\mathcal{C})A_j(\mathcal{C})]$ is too coarse to support the computation of serially concatenated codes, even if an analogue of input-output weight distribution (see e.g. \cite{Divsalar199809}) is introduced. Recall that a linear code is the kernel or image of a linear transformation. Then most kinds of combinations of linear codes can be expressed as a series of two basic operations of linear transformations, namely, the composition (serial concatenation) and the Cartesian product (parallel concatenation).

Motivated by the question above, we provide in this paper a novel property of codes, called second-order weight distributions. From the viewpoint of group actions on sets, the second-order weight is a partition induced by the group of all monomial maps acting on the set $\field{q}^n \times \field{q}^n$, so it is a natural extension of weight, which is a partition induced by the same group acting on $\field{q}^n$. A series of results, parallel to those for weight distributions, is established for second-order weight distributions. In particular, an analogue of MacWilliams identities is proved. Equipped with this new tool, we compute the second-order weight distributions of regular LDPC code ensembles. As easy consequences, we obtain the second moments of weight distributions of regular LDPC code ensembles, which include the results of \cite{Rathi200509, Barak200509, Rathi200609} as special cases. Furthermore, we discuss the application of second-order weight distributions in random coding approach. We compute the second-order weight distributions of the ensembles generated by a so-called $2$-good random generator or parity-check matrix. A $2$-good random matrix is a kind of generalization of the uniformly distributed random matrix and is very useful for solving problems that involve pairwise or triple-wise properties of sequences. We show that the $2$-good property is reflected in the second-order weight distribution, which thus plays a fundamental role in some well-known problems in coding theory and combinatorics. An example of linear intersecting codes is finally provided to illustrate this fact.

The rest of this paper is organized as follows. In Section~\ref{sec:Theory}, we establish the method of second-order weight distributions. In Section~\ref{sec:Applications}, we compute the second-order weight distributions of regular LDPC code ensembles. The application of second-order weight distributions in random coding approach is discussed in Section~\ref{sec:RandomCoding}. Section~\ref{sec:Conclusion} concludes the paper.

In the sequel, the symbols $\nnintegers$, $\complex$, $\symgroup{n}$ denote the set of nonnegative integers, the field of complex numbers, and the group of all permutations on $n$ letters, respectively. The multiplicative subgroup of nonzero elements of $\field{q}$ (resp. $\complex$) is denoted by $\nzfield{q}$ (resp. $\nzcomplex$). A vector in $\field{q}^n$ is typically denoted in the row-vector form $\seq{v} = (v_1, v_2, \cdots, v_n)$. The canonical projection $\pi_i: \field{q}^n \to \field{q}$ is given by $\seq{v} \mapsto v_i$. In general, for an element in a set $A^I$, we adopt a similar notation such as $\seq{v} = (v_i)_{i \in I}$ where $v_i \in A$, and the canonical projection $\pi_i: A^I \to A$ with $i \in I$ is given by $\seq{v} \mapsto v_i$. Given $\seq{u} \in A^I$ and $\seq{v} \in B^I$, if the product $\prod_{i \in I} u_i^{v_i}$ makes sense, we write $\seq{u}^\seq{v}$ as the shortening. For any set $A$ and its subset $B$, the indicator function $1_B: A \to \{0,1\}$ is given by $x\mapsto 1$ for $x\in B$ and $x\mapsto 0$ for $x\not\in B$. When the expression of $B$ is long, we write $1B$ in place of $1_B(x)$. Nonrandom codes are denoted by capital letters, while random codes are denoted by script capital letters. Matrices are denoted by boldface capital letters. By a tilde we mean that a matrix such as $\rmat{A}$ is random. Unless stated otherwise, distinct random elements are assumed to be independent.

\section{Second-Order Weight Distributions}\label{sec:Theory}

In order to find a fundamental property that underlies the second moment of weight distributions, let us first take a close look at the product $A_i(C)A_j(C)$, which may be rewritten as
\[
A_i(C)A_j(C) = \sum_{\seq{u}: \weight(\seq{u}) = i} \sum_{\seq{v}: \weight(\seq{v}) = j} 1\{(\seq{u}, \seq{v}) \in C \times C\}.
\]
Also recall the definition of the weight distribution
\[
A_i(C) = \sum_{\seq{u}: \weight(\seq{u}) = i} 1\{\seq{u} \in C\}.
\]
It is then reasonable to guess that the fundamental property that we are seeking may be a sum of $1\{(\seq{u}, \seq{v}) \in C \times C\}$ over some set of vector pairs. More specifically, let $\mathcal{P}$ be a partition of $\field{q}^n \times \field{q}^n$, and then the quantity
\[
A_P(C, C) \eqdef \sum_{(\seq{u}, \seq{v}) \in P} 1\{(\seq{u}, \seq{v}) \in C \times C\}
\quad \mbox{for $P \in \mathcal{P}$}
\]
gives a kind of property of $C$. Whenever $\mathcal{P}$ is a refinement of the partition $\mathcal{Q} \eqdef \{Q(i,j) \eqdef \{(\seq{u}, \seq{v}): \weight(\seq{u}) = i, \weight(\seq{v}) = j\}\}$, $A_P(C,C)$ can readily yield $A_i(C)A_j(C)$.

One obvious choice of $\mathcal{P}$ is the finest partition of $\field{q}^n \times \field{q}^n$, i.e., $\mathcal{O} \eqdef \{\{(\seq{u}, \seq{v})\}\}_{\seq{u} \in \field{q}^n, \seq{v} \in \field{q}^n}$. However, the partition $\mathcal{O}$ contains so much information that the complexity of induced formulas grows out of control as $n$ increases. On the other hand, as we have shown in Section~\ref{sec:Introduction}, the coarsest partition $\mathcal{Q}$ itself is not qualified, because it contains no enough information. Then our task is now to find an appropriate partition between $\mathcal{O}$ and $\mathcal{Q}$.

A similar story has ever happened on the weight distribution. In order to find the answer, we shall briefly review the reason why the weight distribution is so fundamental.

Let $\rmat{G}$ be the random matrix uniformly distributed over the set $\field{q}^{m \times n}$ of all $m\times n$ matrices over $\field{q}$. It is well known that the linear code ensembles $\{\seq{u}\rmat{G}: \seq{u} \in \field{q}^m\}$ and $\{\seq{u} \in \field{q}^n: \rmat{G}\seq{u}^T = \seq{0}\}$ are both good for channel coding \cite{Barg200209, Gallager196300}. Moreover, the application of $\rmat{G}$ is not confined in channel coding. It also turns out to be good for Slepian-Wolf coding \cite{Csiszar198207}, lossless joint source-channel coding (lossless JSCC) \cite{Yang200904}, and so on.

The success of $\rmat{G}$ in information theory exclusively depends on its fundamental property:
\begin{equation}\label{eq:StrongSCCGood}
P\{F(\seq{u}) = \seq{v}\} = q^{-n}, \qquad \forall \seq{u} \in \field{q}^m \setminus \{\seq{0}\}, \seq{v} \in \field{q}^n
\end{equation}
where $F(\seq{u}) \eqdef \seq{u} \rmat{G}$. In fact, any random linear transformations satisfying \eqref{eq:StrongSCCGood} has the same performance as $F$ for channel coding, lossless JSCC, etc. We may call such random linear transformations \emph{good random linear transformations}.

One important property of good random linear transformations is that both $f \circ F$ and $F \circ g$ are good for any bijective linear transformations $f: \field{q}^n \to \field{q}^n$ and $g: \field{q}^m \to \field{q}^m$. In particular, a good random linear transformation is preserved under a special class of mappings called monomial maps \cite[Sec.~1.7]{Huffman200300}. Let $\seq{c} \in (\nzfield{q})^n$ and $\sigma \in \symgroup{n}$. We define the \emph{monomial map} $\xi_{\seq{c},\sigma}: \field{q}^n \to \field{q}^n$ by
\begin{equation}\label{eq:Permutation1}
\xi_{\seq{c},\sigma}(\seq{v}) \eqdef (c_1 v_{\sigma^{-1}(1)}, c_2 v_{\sigma^{-1}(2)}, \ldots, c_n v_{\sigma^{-1}(n)}).
\end{equation}
Furthermore, we define the \emph{uniform random monomial map} $\Xi_n$ as a random mapping uniformly distributed over the set of all monomial maps. Then given a linear transformation $f: \field{q}^m \to \field{q}^n$, we define the randomization operator
\begin{equation}\label{eq:Operator}
\mathcal{R}(f) \eqdef \Xi_n \circ f \circ \Xi_m.
\end{equation}
Note that, according to our convention, $\Xi_m$ and $\Xi_n$ are independent. It is clear that for any good random linear transformation $F$,
\[
P\{\mathcal{R}(F)(\seq{u}) = \seq{v}\} = q^{-n} \qquad \forall \seq{u} \in \field{q}^m \setminus \{\seq{0}\}, \seq{v} \in \field{q}^n.
\]
This implies that $\mathcal{R}$ does not deteriorate the average performance of the ensemble. Moreover, the new ensemble $\mathcal{R}(F)$ gets larger than $F$ and has more symmetries, which facilitate the analysis of codes. Proceeding with this notion, we may consider such a coding system, where all linear transformations are randomized by independent operators $\mathcal{R}$. It is not a new idea. Both Turbo codes \cite{Berrou199305} and LDPC codes are constructed by this randomization technique, and the analysis of weight distributions always enjoys such code ensembles.

Let $M_n$ be the set of all monomial maps. Then under function composition, $M_n$ forms a group (called a \emph{monomial group}) that acts on $\field{q}^n$.\footnote{A group $G$ is said to act on a set $S$, if there is a function (or action) $G\times S \to S$ such that for all $x \in S$ and $g_1, g_2 \in G$, $ex = x$ and $(g_1g_2)x = g_1(g_2x)$, where $e$ is the identity element of $G$. For a given $x \in S$, the set $\{gx: g \in G\}$ is called the \emph{orbit} of $x$ and is denoted by $\overline{x}$. For different $x, y \in S$ with $\overline{x} \cap \overline{y} \ne \varnothing$, we have $\overline{x} = \overline{y}$, so all sets $\overline{x}$ for $x \in S$ form a partition of $S$. For details, the reader is referred to \cite[Sec.~II.4]{Hungerford197400}.} This notion then establishes the relation between $M_n$ and weight, that is, each set of vectors with the same weight corresponds to exactly one orbit of $M_n$ on $\field{q}^n$. In other words, the weight is nothing but an identification of the orbits of $M_n$ on $\field{q}^n$. This explains why the average weight distribution of Turbo codes and LDPC codes are easier to compute than other codes not randomized by $\mathcal{R}$.

Now that the orbits of $M_n$ on $\field{q}^n$ induce the weight, it is natural to consider the orbits of $M_n$ on $\field{q}^n \times \field{q}^n$ by the action $(\xi, (\seq{u}, \seq{v})) \mapsto (\xi(\seq{u}), \xi(\seq{v}))$. As we shall see later (Lemmas~\ref{le:CartesianProduct} and \ref{le:Permutation} and Theorem~\ref{th:MacWilliams}), these orbits give an appropriate partition of $\field{q}^n\times\field{q}^n$ between $\mathcal{O}$ and $\mathcal{Q}$. Therefore, we choose this partition as the basis for defining the fundamental property of codes. As a natural extension of weight, it will be called the \emph{second-order weight}.\footnote{The second-order weight is not the unique partition satisfying our requirements. If replacing $M_n$ with the group of all permutations of coordinates, i.e., the set of all $\xi_{\seq{1}, \sigma}$ with $\sigma \in \symgroup{n}$, we shall get a finer partition, which also serves our goals and may be called the \emph{second-order complete weight or second-order spectrum}, since it is a natural extension of the complete weight or spectrum (see \cite{Bennatan200403, MacWilliams197211, Yang200904}). However, the second-order complete weight is much more complex than the second-order weight.} Now, it only remains to give a convenient identification of the orbits.

To make things easier to be understood, let us begin with $n=1$. Consider the orbits of $M_1$ (i.e., $\nzfield{q}$) on $\field{q}^2$. There are totally $q+2$ orbits in $\field{q}^2$. We denote by $\sop$ the set of all orbits. It is clear that $\sop$ consists of the zero subspace and $q+1$ one-dimensional subspaces of $\field{q}^2$ with $(0,0)$ excluded. For each orbit $S \in \sop$, we also denote it by $\overline{v}$ for any chosen representative $v \in S$. The standard representative $\rho(S)$ of an orbit $S \in \sop$ is defined as its unique element whose first nonzero component is $1$, or $(0,0)$ for the special orbit $\{(0,0)\}$. For convenience, we also define some special elements and subsets of $\sop$:
\begin{IEEEeqnarray*}{l}
\sope_{00} \eqdef \overline{(0,0)}, \quad \sope_{01} \eqdef \overline{(0,1)}, \quad \sope_{10} \eqdef \overline{(1,0)} \\
\sop_{00} \eqdef \{\sope_{00}\}, \quad \sop_{11} \eqdef \sop\setminus\{\sope_{00}, \sope_{01}, \sope_{10}\}.
\end{IEEEeqnarray*}

\begin{example}
When $q = 3$, we have
\begin{IEEEeqnarray*}{rCl}
\sop
&= &\{\{(0,0)\}, \{(0,1), (0,2)\}, \{(1,0), (2,0)\},\\
& &\{(1,1), (2,2)\}, \{(1,2), (2,1)\}\} \\
\rho(\sop)
&= &\{(0,0), (0,1), (1,0), (1,1), (1,2)\}.
\end{IEEEeqnarray*}
\end{example}

Having introduced the basic notations for identifying every orbits of $M_1$ on $\field{q}^2$, we are now ready to formally define the second-order weight and all related concepts. Our approach simply follows a similar way in which the weight as well as the weight distribution is defined.

\begin{definition}
For any $\seq{u}, \seq{v} \in \field{q}^n$, the \emph{second-order weight} of $(\seq{u}, \seq{v})$ is a $(q+2)$-tuple defined by
\[
\sow(\seq{u}, \seq{v}) \eqdef \left(\sum_{i=1}^n 1_S(u_i, v_i)\right)_{S \in \sop}.
\]
\end{definition}

With this definition, it is easy to verify that the second-order weight parametrizes the orbits of the monomial group $M_n$ on $\field{q}^n \times \field{q}^n$. The next lemma formally states this fact.

\begin{lemma}
If the monomial group $M_n$ acts on $\field{q}^n \times \field{q}^n$ by the action $(\xi, (\seq{u}, \seq{v})) \mapsto (\xi(\seq{u}), \xi(\seq{v}))$, then the orbit of $(\seq{u}, \seq{v}) \in \field{q}^n \times \field{q}^n$ is exactly the set of all vector pairs of second-order weight $\sow(\seq{u}, \seq{v})$.
\end{lemma}

\begin{remark}
As $q=2$, the second-order weight coincides with the well-known joint weight (see e.g. \cite{MacWilliams197211, Dougherty200106}). However, they are different in general. For any $(\seq{u}, \seq{v}) \in \field{q}^n\times \field{q}^n$, the joint weight of $(\seq{u}, \seq{v})$ is a $4$-tuple $(w_{0,0}, w_{1,0}, w_{0,1}, w_{1,1})$ with
\[
w_{a,b} \eqdef \sum_{i=1}^n 1\{\weight(u_i)=a, \weight(v_i)=b\} \qquad \mbox{for $a,b = 0,1$.}
\]
From the viewpoint of group actions on sets, the joint weight is essentially an identification of the orbits of $M_n \times M_n$ on $\field{q}^n \times \field{q}^n$ by the action $((\xi, \zeta), (\seq{u}, \seq{v})) \mapsto (\xi(\seq{u}), \zeta(\seq{v}))$. Since the group $M_n$ can be embedded (as a diagonal subgroup, which is proper for $q \ge 3$) into $M_n \times M_n$ by the monomorphism $\xi \mapsto (\xi, \xi)$, the partition yielded by the second-order weight is a refinement of the partition yielded by the joint weight. Therefore, when $q \ge 3$, the second-order weight provides more information than the joint weight. For example, we can determine whether two vectors are linearly independent by their second-order weight, but not by their joint weight. Suppose that the second-order weight of $(\seq{u}, \seq{v}) \in \field{q}^n \times \field{q}^n$ is $\seq{i} = (i_S)_{S \in \sop}$. Then $\seq{u}$ and $\seq{v}$ are linearly independent if and only if
\[
\sum_{S \in \sop_{00}^c} 1\{i_S > 0\} > 1.
\]
On the other hand, consider the following two pairs of vectors in $\field{3}^3$:
\[
((1, 2, 0), (2, 1, 0)) \mbox{ and } ((1, 1, 0), (2, 1, 0)).
\]
It is clear that they have the same joint weight $(1, 0, 0, 2)$, but the first pair is linearly dependent and the second is linearly independent.
\end{remark}

Next, we proceed to define the second-order weight distribution and the second-order weight enumerator.

\begin{definition}
For any $U, V \subseteq \field{q}^n$, the \emph{second-order weight distribution} of $(U, V)$ is defined by
\[
A_{\seq{i}}(U,V) \eqdef |\{(\seq{u}, \seq{v}) \in U\times V: \sow(\seq{u}, \seq{v}) = \seq{i}\}|
\]
where $\seq{i} \in \mathcal{P}_n \eqdef \{\seq{j} \in \nnintegers^{\sop}: \sum_{S\in \sop} j_S = n\}$.
\end{definition}

\begin{definition}
For any $U, V \subseteq \field{q}^n$, the \emph{second-order weight enumerator} of $(U, V)$ is a polynomial in $q+2$ indeterminates defined by
\begin{IEEEeqnarray*}{rCl}
W_{U,V}(\seq{x})
&\eqdef &\sum_{\seq{u} \in U, \seq{v} \in V} \seq{x}^{\sow(\seq{u}, \seq{v})} \\
&= &\sum_{\seq{i} \in \mathcal{P}_n} A_{\seq{i}}(U,V) \seq{x}^{\seq{i}}
\end{IEEEeqnarray*}
where $\seq{x} \eqdef (x_S)_{S\in\sop}$.
\end{definition}

The next four lemmas give basic properties of the second-order weight distribution.

\begin{lemma}\label{le:CartesianProduct}
Let $U = U_1 \times U_2$ and $V = V_1 \times V_2$, where $U_1, V_1 \subseteq \field{q}^{n_1}$ and $U_2, V_2 \subseteq \field{q}^{n_2}$. Then
\[
W_{U,V}(\seq{x}) = W_{U_1,V_1}(\seq{x}) W_{U_2,V_2}(\seq{x}).
\]
\end{lemma}

\begin{IEEEproof}
\begin{IEEEeqnarray*}{rCl}
W_{U,V}(\seq{x})
&= &\sum_{\seq{u} \in U, \seq{v} \in V} \seq{x}^{\sow(\seq{u},\seq{v})} \\
&= &\sum_{\substack{\seq{u}_1 \in U_1, \seq{v}_1 \in V_1 \\ \seq{u}_2 \in U_2, \seq{v}_2 \in V_2}} \seq{x}^{\sow(\seq{u}_1,\seq{v}_1)} \seq{x}^{\sow(\seq{u}_2,\seq{v}_2)} \\
&= &\sum_{\seq{u}_1 \in U_1, \seq{v}_1 \in V_1} \seq{x}^{\sow(\seq{u}_1,\seq{v}_1)} \sum_{\seq{u}_2 \in U_2, \seq{v}_2 \in V_2} \seq{x}^{\sow(\seq{u}_2,\seq{v}_2)} \\
&= &W_{U_1,V_1}(\seq{x}) W_{U_2,V_2}(\seq{x}).
\end{IEEEeqnarray*}
\end{IEEEproof}

\begin{lemma}\label{le:CompleteSet}
\begin{equation}\label{eq:CompleteSet1}
A_{\seq{i}}(\field{q}^n, \field{q}^n) = {n \choose \seq{i}} (q-1)^{n-i_{\sope_{00}}}
\end{equation}
\begin{equation}\label{eq:CompleteSet2}
W_{\field{q}^n, \field{q}^n}(\seq{x}) = \left[x_{\sope_{00}} + (q-1) \sum_{S\in\sop_{00}^c} x_{S} \right]^n
\end{equation}
where
\begin{equation}
{n \choose \seq{i}} \eqdef \frac{n!}{\prod_{S \in \sop} i_S!}.
\end{equation}
\end{lemma}

\begin{IEEEproof}
It is clear that
\[
W_{\field{q}, \field{q}}(\seq{x}) = x_{\sope_{00}} + (q-1) \sum_{S\in\sop_{00}^c} x_{S}.
\]
This together with Lemma~\ref{le:CartesianProduct} yields \eqref{eq:CompleteSet1} and \eqref{eq:CompleteSet2}.
\end{IEEEproof}

\begin{lemma}\label{le:Permutation}
Let $\seq{c} \in (\nzfield{q})^n$ and $\sigma \in \symgroup{n}$. Then for any $U, V \subseteq \field{q}^n$,
\begin{equation}\label{eq:Permutation2}
A_{\seq{i}}(U, V) = A_{\seq{i}}(\xi_{\seq{c},\sigma}(U), \xi_{\seq{c},\sigma}(V)) \qquad \forall \seq{i} \in \mathcal{P}_n.
\end{equation}
where $\xi_{\seq{c},\sigma}$ is a monomial map defined by \eqref{eq:Permutation1}. Moreover, for any random $\mathcal{U}, \mathcal{V} \subseteq \field{q}^n$,
\begin{equation}\label{eq:Permutation3}
P\{\seq{u} \in \Xi_n(\mathcal{U}), \seq{v} \in \Xi_n(\mathcal{V})\}
= \frac{E[A_{\sow(\seq{u}, \seq{v})}(\mathcal{U}, \mathcal{V})]}{A_{\sow(\seq{u}, \seq{v})}(\field{q}^n, \field{q}^n)}
\end{equation}
for all $\seq{u}, \seq{v} \in \field{q}^n$, where $\Xi_n$ is a uniform random monomial map.
\end{lemma}

\begin{IEEEproof}
Identity~\eqref{eq:Permutation2} clearly holds. As for \eqref{eq:Permutation3}, we note that
\[
P\{\seq{u}' \in \Xi_n(\mathcal{U}), \seq{v}' \in \Xi_n(\mathcal{V})\}
= P\{\seq{u} \in \Xi_n(\mathcal{U}), \seq{v} \in \Xi_n(\mathcal{V})\}
\]
whenever $\sow(\seq{u}', \seq{v}') = \sow(\seq{u}, \seq{v})$. Then we have
\begin{IEEEeqnarray*}{l}
A_{\sow(\seq{u}, \seq{v})}(\field{q}^n, \field{q}^n) P\{\seq{u} \in \Xi_n(\mathcal{U}), \seq{v} \in \Xi_n(\mathcal{V})\} \\
\findent = \sum_{\seq{u}', \seq{v}': \sow(\seq{u}', \seq{v}') = \sow(\seq{u}, \seq{v})} P\{\seq{u}' \in \Xi_n(\mathcal{U}), \seq{v}' \in \Xi_n(\mathcal{V})\} \\
\findent = E\left[\sum_{\seq{u}', \seq{v}': \sow(\seq{u}', \seq{v}') = \sow(\seq{u}, \seq{v})} 1\{\seq{u}' \in \Xi_n(\mathcal{U}), \seq{v}' \in \Xi_n(\mathcal{V})\}\right] \\
\findent = E[A_{\sow(\seq{u}, \seq{v})}(\Xi_n(\mathcal{U}), \Xi_n(\mathcal{V}))].
\end{IEEEeqnarray*}
This combined with \eqref{eq:Permutation2} gives \eqref{eq:Permutation3}.
\end{IEEEproof}

\begin{remark}\label{re:SOIOWD}
Lemma~\ref{le:Permutation} can be further generalized to the case of a mapping randomized by $\mathcal{R}$ defined by \eqref{eq:Operator}. To this end, we need the concept of second-order input-output weight distribution, an analogue of input-output weight distribution. This generalization can facilitate the computation of the second-order weight distribution of serially concatenated codes with all component codes randomized by $\mathcal{R}$. Since Lemma~\ref{le:Permutation} is enough for this paper, we leave this generalization to the reader.
\end{remark}

\begin{lemma}\label{le:SecondMoment}
For $U, V \subseteq \field{q}^n$, the product $W_U(x)W_V(y)$ can be obtained from the second-order weight enumerator $W_{U,V}(\seq{x})$ by the substitution
\begin{equation}\label{eq:SecondMoment1}
x_S \mapsto x^{\weight(\pi_1(\rho(S)))}y^{\weight(\pi_2(\rho(S)))} \qquad \forall S \in \sop
\end{equation}
where $\pi_i$ ($i = 1, 2$) is the canonical projection $\field{q}^2 \to \field{q}$ given by $(v_1, v_2) \mapsto v_i$. As a consequence, we have
\begin{equation}\label{eq:SecondMoment2}
A_j(U) A_k(V) = \sum_{l=0}^{\min\{j,k\}} \sum_{\substack{\sum_{S \in \sop_{11}} i_S = l \\ i_{\sope_{10}}=j-l, i_{\sope_{01}}=k-l}} A_{\seq{i}}(U,V).
\end{equation}
\end{lemma}

The proof is left to the reader. Note that $A_j(U) A_k(V) = \coef(W_U(x)W_V(y), x^jy^k)$.

One of the most famous results in coding theory is the MacWilliams identities \cite{MacWilliams196300}. Now, we shall derive an analogue of MacWilliams identities for the second-order weight distribution.

\begin{theorem}\label{th:MacWilliams}
For any $V \subseteq \field{q}^n$, we define the \emph{orthogonal set} $\orth{V}$ by
\begin{equation}\label{eq:MacWilliams1}
\orth{V} \eqdef \left\{\seq{v}' \in \field{q}^n: \seq{v} \cdot \seq{v}' \eqdef \sum_{i=1}^n v_i v_i' = 0 \mbox{ for all $\seq{v} \in V$}\right\}.
\end{equation}
Then for any subspaces $U, V \subseteq \field{q}^n$,
\begin{equation}\label{eq:MacWilliams2}
W_{\orth{U}, \orth{V}}(\seq{x}) = \frac{1}{|U||V|} W_{U,V}(\seq{x} \mat{K})
\end{equation}
where $\mat{K}$ is a $(q+2)\times (q+2)$ matrix $(\matentry{K}_{S,T})_{S,T \in \sop}$ defined by
\begin{subnumcases}{\matentry{K}_{S,T} \eqdef\label{eq:MacWilliams3}}
|S|, &$T \subseteq \orth{S}$ \\
-1, &$T \not\subseteq \orth{S}$.
\end{subnumcases}
\end{theorem}

\begin{IEEEproof}
Since $U$ and $V$ are subspaces of $\field{q}^n$, it follows from Lemma~\ref{le:MacWilliams1} that
\begin{IEEEeqnarray*}{l}
W_{\orth{U}, \orth{V}}(\seq{x}) \\
\findent = \sum_{\seq{u}', \seq{v}' \in \field{q}^n} 1_{\orth{U}}(\seq{u}') 1_{\orth{V}}(\seq{v}') \seq{x}^{\sow(\seq{u}', \seq{v}')} \\
\findent= \frac{1}{|U||V|} \sum_{\seq{u} \in U, \seq{v} \in V} \sum_{\seq{u}', \seq{v}' \in \field{q}^n} \chi(\seq{u} \cdot \seq{u}' + \seq{v} \cdot \seq{v}') \seq{x}^{\sow(\seq{u}',\seq{v}')}.
\end{IEEEeqnarray*}
Applying Lemma~\ref{le:MacWilliams2} then gives
\[
W_{\orth{U}, \orth{V}}(\seq{x}) = \frac{1}{|U||V|} \sum_{\seq{u} \in U, \seq{v} \in V} (\seq{x} \mat{K})^{\sow(\seq{u}, \seq{v})}
\]
as desired.
\end{IEEEproof}

\begin{remark}
The set $\field{q}^2$, as a direct product of $\field{q}$, is a Frobenius ring because $\field{q}$ is a Frobenius ring and the class of Frobenius rings is closed under finite direct products of rings.\footnote{A ring $R$ is said to be a \emph{Frobenius ring} if there exists a group homomorphism $f: (R, +) \to \nzcomplex$ (character of $(R, +)$), whose kernel contains no nonzero left or right ideal of $R$. Such a homomorphism is called a \emph{generating character} of $R$. The reader is referred to \cite[Sec.~16]{Lam1999} for background information on Frobenius rings.} Consequently, Theorem~\ref{th:MacWilliams} can be regarded as a consequence of the generalized MacWilliams identities for linear codes over finite Frobenius rings \cite{Wood199903, Honold200100}.
\end{remark}

\section{Second-Order Weight Distributions of Regular LDPC Code Ensembles}\label{sec:Applications}

Equipped with the tool established in Section~\ref{sec:Theory}, we proceed to compute the second-order weight distributions of regular LDPC code ensembles.

At first, we compute the second-order weight distributions of two simple codes, the single symbol repetition code and the single symbol check code.

\begin{definition}
A \emph{single symbol repetition map} $\repmap_c: \field{q} \to \field{q}^c$ with the parameter $c$ is given by $v \mapsto (v, v, \cdots, v)$. The image of $\repmap_c$ is called a \emph{single symbol repetition code}, which we denote by $\repcode_c$.
\end{definition}

\begin{lemma}\label{le:Repetition}
For the single symbol repetition code $\repcode_c$,
\[
W_{\repcode_c, \repcode_c}(\seq{x}) = x_{S_{00}}^c + (q-1) \sum_{S \in \sop_{00}^c} x_S^c.
\]
\end{lemma}

The proof is left to the reader.

\begin{definition}
A \emph{single symbol check map} $\chkmap_d: \field{q}^d \to \field{q}$ with the parameter $d$ is given by $\seq{v} \mapsto \sum_{i=1}^d v_i$. The kernel of $\chkmap_d$ is called a \emph{single symbol check code}, which we denote by $\chkcode_d$.
\end{definition}

\begin{lemma}\label{le:Check}
For the single symbol check code $\chkcode_d$,
\begin{IEEEeqnarray*}{rCl}
W_{\chkcode_d,\chkcode_d}(\seq{x})
&= &\frac{1}{q^2} \left[\left( \sum_{S \in \sop} x_{S} \matentry{K}_{S,\sope_{00}} \right)^d \right. \\
& &\breakop{+} \left.(q-1) \sum_{T \in \sop_{00}^c} \left( \sum_{S \in \sop} x_{S} \matentry{K}_{S,T} \right)^d \right]
\end{IEEEeqnarray*}
where $\matentry{K}_{S,T}$ is defined by \eqref{eq:MacWilliams3}.
\end{lemma}

\begin{IEEEproof}
Use Theorems~\ref{th:MacWilliams} and Lemma~\ref{le:Repetition} with $\chkcode_d = \orth{(\repcode_d)}$.
\end{IEEEproof}

\begin{example}
When $q=3$, Lemma~\ref{le:Check} gives
\begin{IEEEeqnarray*}{l}
W_{\chkcode_d,\chkcode_d}(\seq{x})\\
\findent = \frac{1}{9} \left[\left( x_{\overline{(0,0)}} + 2x_{\overline{(0,1)}} + 2x_{\overline{(1,0)}} + 2x_{\overline{(1,1)}} + 2x_{\overline{(1,2)}} \right)^d \right. \\
\findent\nullrel \breakop{+} 2 \left( x_{\overline{(0,0)}} - x_{\overline{(0,1)}} + 2x_{\overline{(1,0)}} - x_{\overline{(1,1)}} - x_{\overline{(1,2)}} \right)^d \\
\findent\nullrel \breakop{+} 2 \left( x_{\overline{(0,0)}} +2 x_{\overline{(0,1)}} - x_{\overline{(1,0)}} - x_{\overline{(1,1)}} - x_{\overline{(1,2)}} \right)^d \\
\findent\nullrel \breakop{+} 2 \left( x_{\overline{(0,0)}} - x_{\overline{(0,1)}} - x_{\overline{(1,0)}} - x_{\overline{(1,1)}} + 2x_{\overline{(1,2)}} \right)^d \\
\findent\nullrel \breakop{+} \left.2 \left( x_{\overline{(0,0)}} - x_{\overline{(0,1)}} - x_{\overline{(1,0)}} + 2x_{\overline{(1,1)}} - x_{\overline{(1,2)}} \right)^d \right].
\end{IEEEeqnarray*}
\end{example}

We are now ready to compute the second-order weight distributions of regular LDPC code ensembles. There are a few kinds of regular LDPC code ensembles \cite{Litsyn200204}. We shall consider here two typical regular LDPC code ensembles. For convenience, we denote by $\chkmap_{d,n}$ (resp. $\repmap_{c,n}$) the $n$-fold Cartesian product of $\chkmap_d$ (resp. $\repmap_c$).

The first ensemble of regular LDPC codes is due to Gallager \cite{Gallager196300}. Though it is only known as a binary regular LDPC code ensemble, its extension to a finite field is immediate.

\begin{definition}\label{df:LDPCA}
Let $c$, $d$, and $n$ be positive integers such that $d$ divides $n$. Let $\ldmapA_{d,n}: \field{q}^n \to \field{q}^{n/d}$ be a random linear transformation defined by
\[
\ldmapA_{d,n}(\seq{v}) \eqdef \chkmap_{d,n/d}(\Xi_{n}(\seq{v})).
\]
The \emph{regular LDPC code ensemble I}, which we denote by $\ldcodeA_{c,d,n}$, is defined as the intersection of $c$ independent copies of the kernel of $\ldmapA_{d,n}$.
\end{definition}

According to this definition, $\ldcodeA_{c,d,n}$ is the solution space of the random equations
\[
H^{(i)}(\seq{v}) = 0, \qquad i=1, 2, \ldots, c
\]
where $H^{(i)}$ is the $i$th independent copy of $\ldmapA_{d,n/d}$. In other words, the parity-check matrix of $\ldcodeA_{c,d,n}$ consists of $c$ submatrices, each being an independent copy of the transformation matrix of $\ldmapA_{d,n/d}$ (with input vectors in column-vector form). Since the transformation matrix of $\ldmapA_{d,n/d}$ contains exactly one nonzero entry in each column and $d$ nonzero entries in each row, the resulting parity-check matrix is a $(nc/d)$-by-$n$ random sparse matrix with $c$ nonzero entries in each column and $d$ nonzero entries in each row, which motivates the term ``regular low-density parity-check code ensemble''.

\begin{theorem}\label{th:LDPCA}
For the regular LDPC code ensemble $\ldcodeA_{c,d,n}$,
\[
E[A_{\seq{i}}(\ldcodeA_{c,d,n}, \ldcodeA_{c,d,n})]
= \frac{\left[\coef([W_{\chkcode_d,\chkcode_d}(\seq{x})]^{n/d}, \seq{x}^{\seq{i}})\right]^c}{\left[{n \choose \seq{i}} (q-1)^{n-i_{\sope_{00}}}\right]^{c-1}}
\]
where $\seq{i} \in \mathcal{P}_n$.
\end{theorem}

\begin{IEEEproof}
For any $\seq{u}, \seq{v} \in \field{q}^n$ with $\sow(\seq{u}, \seq{v}) = \seq{i}$,
\begin{IEEEeqnarray*}{rCl}
P\{\seq{u}, \seq{v} \in \ker\ldmapA_{d,n}\}
&\eqvar{(a)} &P\{\seq{u}, \seq{v} \in \Xi_{n}((\chkcode_d)^{n/d})\} \\
&\eqvar{(b)} &\frac{A_{\seq{i}}((\chkcode_d)^{n/d}, (\chkcode_d)^{n/d})}{A_{\seq{i}}(\field{q}^{n}, \field{q}^{n})} \\
&\eqvar{(c)} &\frac{\coef([W_{\chkcode_d,\chkcode_d}(\seq{x})]^{n/d}, \seq{x}^{\seq{i}})}{{n \choose \seq{i}} (q-1)^{n-i_{\sope_{00}}}}
\end{IEEEeqnarray*}
where (a) follows from Definition~\ref{df:LDPCA}, (b) from Lemma~\ref{le:Permutation}, (c) follows from Lemmas~\ref{le:CartesianProduct} and \ref{le:CompleteSet}. This together with the identity
\begin{IEEEeqnarray*}{rCl}
E[A_{\seq{i}}(\ldcodeA_{c,d,n}, \ldcodeA_{c,d,n})]
&= &\sum_{\seq{u},\seq{v}: \sow(\seq{u},\seq{v}) = \seq{i}} P\{\seq{u}, \seq{v} \in \ldcodeA_{c,d,n}\} \\
&= &\sum_{\seq{u},\seq{v}: \sow(\seq{u},\seq{v}) = \seq{i}} \left(P\{\seq{u}, \seq{v} \in \ker \ldmapA_{d,n}\}\right)^c
\end{IEEEeqnarray*}
establishes the theorem.
\end{IEEEproof}

The second ensemble of regular LDPC codes is the \emph{regular bipartite graph ensemble} suggested by \cite{Luby200102, Richardson200102, Bennatan200403}.

\begin{definition}\label{df:LDPCB}
Let $c$, $d$, and $n$ be positive integers such that $d$ divides $cn$. Let $\ldmapB_{c,d,n}: \field{q}^n \to \field{q}^{cn/d}$ be a random linear transformation defined by
\[
\ldmapB_{c,d,n}(\seq{v}) \eqdef \chkmap_{d,cn/d}(\Xi_{cn}(\repmap_{c,n}(\seq{v}))).
\]
The \emph{regular LDPC code ensemble II}, which we denote by $\ldcodeB_{c,d,n}$, is defined as the kernel of $\ldmapB_{c,d,n}$.%
\footnote{If regarding $\repmap_{c,n}$ as $n$ variable nodes (each with $c$ sockets) and $\chkmap_{d,cn/d}$ as $cn/d$ check nodes (each with $d$ sockets), we immediately obtain the well-known bipartite graph model, where the connections between variable nodes and check nodes are given by the uniform random monomial map $\Xi_{cn}$.}
\end{definition}

\begin{theorem}\label{th:LDPCB}
For the regular LDPC code ensemble $\ldcodeB_{c,d,n}$,
\[
E[A_{\seq{i}}(\ldcodeB_{c,d,n}, \ldcodeB_{c,d,n})]
= \frac{{n \choose \seq{i}} \coef([W_{\chkcode_d,\chkcode_d}(\seq{x})]^{cn/d}, \seq{x}^{c\seq{i}})}{{cn \choose c\seq{i}} (q-1)^{(c-1)(n-i_{\sope_{00}})}}
\]
where $\seq{i} \in \mathcal{P}_n$ and $c\seq{i} \eqdef (ci_S)_{S\in\sop}$.
\end{theorem}

\begin{IEEEproof}
For any $\seq{u}, \seq{v} \in \field{q}^n$ with $\sow(\seq{u}, \seq{v}) = \seq{i}$,
\begin{IEEEeqnarray*}{l}
P\{\seq{u}, \seq{v} \in \ldcodeB_{c,d,n}\} \\
\findent \eqvar{(a)} P\{\repmap_{c,n}(\seq{u}), \repmap_{c,n}(\seq{v}) \in \Xi_{cn}((\chkcode_d)^{cn/d})\} \\
\findent \eqvar{(b)} \frac{A_{c\seq{i}}((\chkcode_d)^{cn/d}, (\chkcode_d)^{cn/d})}{A_{c\seq{i}}(\field{q}^{cn}, \field{q}^{cn})} \\
\findent \eqvar{(c)} \frac{\coef([W_{\chkcode_d,\chkcode_d}(\seq{x})]^{cn/d}, \seq{x}^{c\seq{i}})}{{cn \choose c\seq{i}} (q-1)^{c(n-i_{\sope_{00}})}}
\end{IEEEeqnarray*}
where (a) follows from Definition~\ref{df:LDPCB}, (b) from Lemma~\ref{le:Permutation}, (c) follows from Lemmas~\ref{le:CartesianProduct} and \ref{le:CompleteSet}. This together with the identity
\[
E[A_{\seq{i}}(\ldcodeB_{c,d,n}, \ldcodeB_{c,d,n})]
= \sum_{\seq{u},\seq{v}: \sow(\seq{u},\seq{v}) = \seq{i}} P\{\seq{u}, \seq{v} \in \ldcodeB_{c,d,n}\}
\]
establishes the theorem.
\end{IEEEproof}

\begin{remark}
When $q=2$, Lemma~\ref{le:Check} and Theorems~\ref{th:LDPCA} and \ref{th:LDPCB} give
\[
E[A_{\seq{i}}(\ldcodeA_{c,d,n}, \ldcodeA_{c,d,n})]
= \frac{\left[\coef((g_d(\seq{x}))^{n/d}, \seq{x}^{\seq{i}})\right]^c}{{n \choose \seq{i}}^{c-1}}
\]
and
\[
E[A_{\seq{i}}(\ldcodeB_{c,d,n}, \ldcodeB_{c,d,n})]
= \frac{{n \choose \seq{i}} \coef((g_d(\seq{x}))^{cn/d}, \seq{x}^{c\seq{i}})}{{cn \choose c\seq{i}}}
\]
where
\begin{IEEEeqnarray*}{rCl}
g_d(\seq{x})
&\eqdef &\frac{1}{4} \left[\left(x_{\overline{(0,0)}} + x_{\overline{(0,1)}} + x_{\overline{(1,0)}} + x_{\overline{(1,1)}}\right)^d\right. \\
& &\breakop{+} \left(x_{\overline{(0,0)}} - x_{\overline{(0,1)}} + x_{\overline{(1,0)}} - x_{\overline{(1,1)}}\right)^d \\
& &\breakop{+} \left(x_{\overline{(0,0)}} + x_{\overline{(0,1)}} - x_{\overline{(1,0)}} - x_{\overline{(1,1)}}\right)^d \\
& &\breakop{+} \left.\left(x_{\overline{(0,0)}} - x_{\overline{(0,1)}} - x_{\overline{(1,0)}} + x_{\overline{(1,1)}}\right)^d\right].
\end{IEEEeqnarray*}
Furthermore, Lemma~\ref{le:SecondMoment} shows that
\begin{IEEEeqnarray*}{l}
E[A_j(\ldcodeA_{c,d,n}) A_k(\ldcodeA_{c,d,n})] \\
\findent = \sum_{l=0}^{\min\{j,k\}} \Bigg[ {n \choose l\;j-l\;k-l\;n-j-k+l}^{-(c-1)} \\
\findent \nullrel \breakop{\times} \left(\coef((g_d(\seq{x}))^{n/d}, x_{\overline{(0,0)}}^{n-j-k+l} x_{\overline{(0,1)}}^{k-l} x_{\overline{(1,0)}}^{j-l} x_{\overline{(1,1)}}^{l})\right)^c \Bigg]
\end{IEEEeqnarray*}
and
\begin{IEEEeqnarray*}{l}
E[A_j(\ldcodeB_{c,d,n}) A_k(\ldcodeB_{c,d,n})] \\
\findent = \sum_{l=0}^{\min\{j,k\}} \Bigg( \frac{{n \choose l\;j-l\;k-l\;n-j-k+l}}{{cn \choose cl\;c(j-l)\;c(k-l)\;c(n-j-k+l)}} \\
\findent \nullrel \breakop{\times} \coef((g_d(\seq{x}))^{cn/d}, x_{\overline{(0,0)}}^{c(n-j-k+l)} x_{\overline{(0,1)}}^{c(k-l)} x_{\overline{(1,0)}}^{c(j-l)} x_{\overline{(1,1)}}^{cl}) \Bigg).
\end{IEEEeqnarray*}
The second formula with $j=k$ coincides with the second-moment formula of binary regular LDPC codes given by \cite{Rathi200509, Barak200509, Rathi200609}.
\end{remark}

\section{Applications in Random Coding Approach}\label{sec:RandomCoding}

As discussed in Section~\ref{sec:Theory}, the uniformly distributed random $m\times n$ matrix $\rmat{G}$ plays an important role in information theory because of the property that $\seq{u}\rmat{G}$ is uniformly distributed over $\field{q}^n$ for every nonzero $\seq{u}\in \field{q}^m$. But it is not the end of the story. Theorem~\ref{th:RLC2} shows that $\mat{U}\rmat{G}$ is uniformly distributed over $\field{q}^{m\times n}$ for every invertible matrix $\mat{U} \in \field{q}^{m\times m}$. In particular, if $m \ge 2$, the product $\mat{U}\rmat{G}$ is uniformly distributed over $\field{q}^{2\times n}$ for every matrix $\mat{U} \in \field{q}^{2\times m}$ of rank $2$. In other words, for any two linearly independent vectors $\seq{u}, \seq{v} \in \field{q}^m$, the random vectors $\seq{u}\rmat{A}, \seq{v}\rmat{A}$ are uniformly distributed and independent (see Corollary~\ref{co:RLC}). In this section, we shall show that this property is reflected in the second-order weight distribution, which thus plays an important role in some well-known problems in coding theory and combinatorics.

At first, for positive integers $m,n,k$ with $k \le \min\{m,n\}$, a random $m\times n$ matrix $\rmat{A}$ is said to be $k$-good if $\mat{U}\rmat{A}$ is uniformly distributed over $\field{q}^{k\times n}$ for every matrix $\mat{U} \in \field{q}^{k\times m}$ of rank $k$.\footnote{In the sequel, when speaking of a $k$-good random $m\times n$ matrix, we shall tacitly assume that $k \le \min\{m,n\}$.} In this paper, we are only concerned with the cases $k=1, 2$. It is clear that the uniformly distributed random $m\times n$ matrix $\rmat{G}$ with $m,n \ge 2$ is $2$-good and $2$-goodness implies $1$-goodness. However, a $1$-good random matrix is not necessarily $2$-good and there are also other $2$-good random matrices than $\rmat{G}$. The next example illustrates these two facts.

\begin{example}
Consider the matrix space $\field{2}^{3\times 3}$, which is identified with $\field{8}^3$ by viewing the columns of $\mat{A} \in \field{2}^{3\times 3}$ as coordinate vectors relative to $1, \alpha, \alpha^2$, where $\alpha^3 + \alpha + 1 = 0$. The random matrix uniformly distributed over the set
\[
\matset{A}_1 \eqdef \{x (1, \alpha, \alpha^2): x \in \field{8}\}
\]
is $1$-good but not $2$-good. The random matrix uniformly distributed over the set
\[
\matset{A}_2 \eqdef \{x (1, \alpha, \alpha^2) + y (1, \alpha^2, \alpha^4): x, y \in \field{8}\}
\]
is $2$-good and it is clear that $\matset{A}_2$ is a proper subset of $\field{2}^{3\times 3}$. In fact, both $\matset{A}_1$ and $\matset{A}_2$ are maximum-rank-distance (MRD) codes \cite{Delsarte197800, Gabidulin198501, Roth199103}. The reader is referred to \cite{Yang201011} for a detailed investigation of the relation between $k$-good random matrices and MRD codes.
\end{example}

Next, let us compute the second-order weight distribution of linear code ensembles generated by a $2$-good random generator or parity-check matrix.


\begin{theorem}\label{th:RLC}
Let $\rmat{A}$ be a $2$-good random $m\times n$ matrices with $m < n$. Let $\rlccode_{m,n}$ and $\rlccodeB_{m,n}$ be two random linear code ensembles generated by the generator matrix $\rmat{A}$ and the parity-check matrix $\rmat{A}$, respectively. Then we have
\begin{IEEEeqnarray*}{l}
E\left[\frac{W_{\rlccode_{m,n},\rlccode_{m,n}}(\seq{x})}{|\rlccode_{m,n}|^2}\right] \\
\findent = \frac{1}{q^{2m}} x_{\sope_{00}}^n + \frac{q^m-1}{q^{2m}} \sum_{S \in \sop_{00}^c} \left(\frac{1}{q} x_{\sope_{00}} + \frac{q-1}{q} x_{S}\right)^n \\
\findent \nullrel \breakop{+} \frac{(q^m-1)(q^m-q)}{q^{2m}} \left(\frac{1}{q^2} x_{\sope_{00}} + \frac{q-1}{q^2} \sum_{S \in \sop_{00}^c} x_S\right)^n
\end{IEEEeqnarray*}
and
\begin{IEEEeqnarray*}{rCl}
E[W_{\rlccodeB_{m,n},\rlccodeB_{m,n}}(\seq{x})]
&= &\frac{(q^m-1)(q^m-q)}{q^{2m}} x_{\sope_{00}}^n \\
& &\breakop{+} \frac{q^m-1}{q^{2m}} \sum_{S \in \sop_{00}^c} \left(x_{\sope_{00}} + (q-1)x_{S}\right)^n \\
& &\breakop{+} \frac{1}{q^{2m}} \left(x_{\sope_{00}} + (q-1) \sum_{S \in \sop_{00}^c} x_S\right)^n.
\end{IEEEeqnarray*}
\end{theorem}

\begin{IEEEproof}
Since $\rlccode_{m,n} = \{\seq{u}\rmat{A}: \seq{u} \in \field{q}^m\}$, we have
\begin{IEEEeqnarray*}{l}
E\left[\frac{W_{\rlccode_{m,n},\rlccode_{m,n}}(\seq{x})}{|\rlccode_{m,n}|^2}\right] \\
\findent = E\left[\frac{\sum_{\seq{v}, \seq{v}' \in \rlccode_{m,n}} \seq{x}^{\sow(\seq{v}, \seq{v}')}}{|\rlccode_{m,n}|^2}\right] \\
\findent = \frac{1}{q^{2m}} E\left[\sum_{\seq{u}, \seq{u}' \in \field{q}^m} \seq{x}^{\sow(\seq{u}\rmat{A}, \seq{u}'\rmat{A})}\right] \\
\findent = \frac{1}{q^{2m}} E\Bigg[ \seq{x}^{\sow(\seq{0}, \seq{0})} + \sum_{\seq{u}=\seq{0}, \seq{u}'\ne\seq{0}} \seq{x}^{\sow(\seq{u}\rmat{A}, \seq{u}'\rmat{A})} \\
\findent \nullrel \breakop{+} \sum_{\seq{u}\ne\seq{0}, \seq{u}'=\seq{0}} \seq{x}^{\sow(\seq{u}\rmat{A}, \seq{u}'\rmat{A})} + \sum_{a \in \nzfield{q}} \sum_{\seq{u} \ne 0} \seq{x}^{\sow(\seq{u}\rmat{A}, a\seq{u}\rmat{A})} \\
\findent \nullrel \breakop{+} \sum_{\seq{u}, \seq{u}'\text{ is linearly independent}} \seq{x}^{\sow(\seq{u}\rmat{A}, \seq{u}'\rmat{A})} \Bigg] \\
\findent = \frac{1}{q^{2m}} \Bigg( x_{\sope_{00}}^n + \sum_{\seq{u}' \ne 0} E\left[ \seq{x}^{\sow(\seq{0}, \seq{u}'\rmat{A})}\right] \\
\findent \nullrel \breakop{+} \sum_{\seq{u} \ne 0} E\left[ \seq{x}^{\sow(\seq{u}\rmat{A}, \seq{0})}\right] + \sum_{a \in \nzfield{q}} \sum_{\seq{u} \ne 0} E\left[ \seq{x}^{\sow(\seq{u}\rmat{A}, a\seq{u}\rmat{A})} \right] \\
\findent \nullrel \breakop{+} \sum_{\seq{u}, \seq{u}'\text{ is linearly independent}} E\left[\seq{x}^{\sow(\seq{u}\rmat{A}, \seq{u}'\rmat{A})}\right] \Bigg) \\
\findent \eqvar{(a)} \frac{1}{q^{2m}} \Bigg( x_{\sope_{00}}^n + (q^m-1) \sum_{\seq{v}' \in \field{q}^n} q^{-n} \seq{x}^{\sow(\seq{0}, \seq{v}')} \\
\findent \nullrel \breakop{+} (q^m-1) \sum_{\seq{v} \in \field{q}^n} q^{-n} \seq{x}^{\sow(\seq{v}, \seq{0})} \\
\findent \nullrel \breakop{+} (q^m-1) \sum_{a\in \nzfield{q}} \sum_{\seq{v} \in \field{q}^n} q^{-n} \seq{x}^{\sow(\seq{v}, a\seq{v})} \\
\findent \nullrel \breakop{+} (q^m-1)(q^m-q) \sum_{\seq{v}, \seq{v}' \in \field{q}^n} q^{-2n} \seq{x}^{\sow(\seq{v}, \seq{v}')} \Bigg) \\
\findent = \frac{1}{q^{2m}} \Bigg[ x_{\sope_{00}}^n + (q^m-1) \Bigg( \sum_{v' \in \field{q}} q^{-1} x_{\overline{(0, v')}} \Bigg)^n \\
\findent \nullrel \breakop{+} (q^m-1) \Bigg( \sum_{v \in \field{q}} q^{-1} x_{\overline{(v, 0)}} \Bigg)^n \\
\findent \nullrel \breakop{+} (q^m-1) \sum_{a\in \nzfield{q}} \Bigg( \sum_{v \in \field{q}} q^{-1} x_{\overline{(v, av)}} \Bigg)^n \\
\findent \nullrel \breakop{+} (q^m-1)(q^m-q) \Bigg( \sum_{v, v' \in \field{q}} q^{-2} x_{\overline{(v, v')}} \Bigg)^n \Bigg] \\
\findent = \frac{1}{q^{2m}} \Bigg[ x_{\sope_{00}}^n + (q^m-1) \sum_{S \in \sop_{00}^c} \left(\frac{1}{q} x_{\sope_{00}} + \frac{q-1}{q} x_{S}\right)^n \\
\findent \nullrel \breakop{+} (q^m-1)(q^m-q) \Bigg( \frac{1}{q^2} x_{\sope_{00}} + \frac{q-1}{q^2} \sum_{S \in \sop_{00}^c} x_{S} \Bigg)^n \Bigg]
\end{IEEEeqnarray*}
where (a) follows from the $2$-goodness of $\rmat{A}$.

Next, since $\rlccodeB_{m,n} \eqdef \{\seq{v}\in \field{q}^n: \rmat{A}\seq{v}^T = \seq{0}\}$, it is clear that $\rlccodeB_{m,n} = \orth{(\rlccode_{m,n})}$, and then Theorem~\ref{th:MacWilliams} shows that
\begin{IEEEeqnarray*}{rCl}
E[W_{\rlccodeB_{m,n},\rlccodeB_{m,n}}(\seq{x})]
&= &E\left[\frac{W_{\rlccode_{m,n},\rlccode_{m,n}}(\seq{x} \mat{K})}{|\rlccode_{m,n}|^2}\right] \\
&= &\frac{(q^m-1)(q^m-q)}{q^{2m}} x_{\sope_{00}}^n \\
& &\breakop{+} \frac{q^m-1}{q^{2m}} \sum_{S \in \sop_{00}^c} \left(x_{\sope_{00}} + (q-1)x_{S}\right)^n \\
& &\breakop{+} \frac{1}{q^{2m}} \left(x_{\sope_{00}} + (q-1) \sum_{S \in \sop_{00}^c} x_S\right)^n.
\end{IEEEeqnarray*}
\end{IEEEproof}

\begin{remark}
Note that the size of $\rlccode_{m,n}$ is random and it may be less than $q^m$. For this reason, we give in Theorem~\ref{th:RLC} the expectation of the ratio $W_{\rlccode_{m,n},\rlccode_{m,n}}(\seq{x}) / |\rlccode_{m,n}|^2$ instead of $W_{\rlccode_{m,n},\rlccode_{m,n}}(\seq{x})$. Theorem~\ref{th:RLC} and Lemma~\ref{le:SecondMoment} show that
\begin{IEEEeqnarray*}{l}
E\left[\frac{W_{\rlccode_{m,n}}(x) W_{\rlccode_{m,n}}(y)}{|\rlccode_{m,n}|^2}\right] \\
\findent = \frac{1}{q^{2m}} + \frac{q^m-1}{q^{2m}} \left(\frac{1}{q} + \frac{q-1}{q} x\right)^n \\
\findent \nullrel \breakop{+} \frac{q^m-1}{q^{2m}} \left(\frac{1}{q} + \frac{q-1}{q} y\right)^n \\
\findent \nullrel \breakop{+} \frac{(q^m-1)(q-1)}{q^{2m}} \left(\frac{1}{q} + \frac{q-1}{q} xy\right)^n \\
\findent \nullrel \breakop{+} \frac{(q^m-1)(q^m-q)}{q^{2m}} \bigg(\frac{1}{q} + \frac{q-1}{q} x\bigg)^n \bigg(\frac{1}{q} + \frac{q-1}{q} y\bigg)^n.
\end{IEEEeqnarray*}
This is equivalent to the formula \cite[Eq. (2)]{Blinovsky200912}, in which a linear code is allowed to contain duplicated codewords, so that $E[(A_0(\rlccode_{m,n}))^2] > 1$.
\end{remark}

Now let us show that the $2$-good property is reflected in the second-order weight distribution given by Theorem~\ref{th:RLC}.

\begin{proposition}\label{pr:PairwiseIndependence}
Let $\mathcal{D}$ be a random linear code having the same average second-order weight distribution as $\rlccodeB_{m,n}$. Then we have
\[
P\{\seq{u} \in \Xi_n(\mathcal{D})\} = q^{-m} \qquad \mbox{for any nonzero $\seq{u} \in \field{q}^n$}
\]
and
\[
P\{\seq{u} \in \Xi_n(\mathcal{D}), \seq{v} \in \Xi_n(\mathcal{D})\} = q^{-2m}
\]
for any linearly independent $\seq{u}, \seq{v} \in \field{q}^n$.%
\footnote{There is also a similar result for $\rlccode_{m,n}$, but it will be slightly complicated since we only know the expectation of the ratio $W_{\rlccode_{m,n},\rlccode_{m,n}}(\seq{x}) / |\rlccode_{m,n}|^2$ and the fact that the coding rate of $\rlccode_{m,n}$ is approximately $m/n$ with high probability for sufficiently large $n$.}
\end{proposition}

\begin{proof}
From Lemma~\ref{le:Permutation} and Theorem~\ref{th:RLC}, it follows that for any nonzero $\seq{u}\in \field{q}^n$,
\begin{IEEEeqnarray*}{rCl}
P\{\seq{u} \in \Xi_n(\mathcal{D})\}
&= &P\{\seq{u} \in \Xi_n(\mathcal{D}), \seq{0} \in \Xi_n(\mathcal{D})\} \\
&= &\frac{E[A_{\sow(\seq{u},\seq{0})}(\rlccodeB_{m,n}, \rlccodeB_{m,n})]}{A_{\sow(\seq{u},\seq{0})}(\field{q}^n, \field{q}^n)} \\
&= &q^{-m}.
\end{IEEEeqnarray*}
Similarly, for any linearly independent $\seq{u}, \seq{v} \in \field{q}^n$, it follows from Lemma~\ref{le:Permutation} and Theorem~\ref{th:RLC} that 
\begin{IEEEeqnarray*}{rCl}
P\{\seq{u} \in \Xi_n(\mathcal{D}), \seq{v} \in \Xi_n(\mathcal{D})\}
&= &\frac{E[A_{\sow(\seq{u},\seq{v})}(\rlccodeB_{m,n}, \rlccodeB_{m,n})]}{A_{\sow(\seq{u},\seq{v})}(\field{q}^n, \field{q}^n)} \\
&= &q^{-2m}.
\end{IEEEeqnarray*}
\end{proof}

Clearly, the second identity in Proposition~\ref{pr:PairwiseIndependence} reflects the $2$-good property. The $2$-good property is very useful for solving problems that involve pairwise or triple-wise properties of sequences (see \cite{Yang201011}). Now based on Proposition~\ref{pr:PairwiseIndependence}, we shall provide an example of linear intersecting codes to show the fundamental position and potential application of second-order weight distribution in some problems in coding theory and combinatorics.

A pair of vectors in $\field{q}^n$ is said to be \emph{intersecting} if there is at least one position $i$ such that their $i$th components are both nonzero. A linear code is said to be \emph{intersecting} if its any two linearly independent codewords intersect. Recall that linear intersecting codes has a close relation to many problems in combinatorics, such as separating systems \cite{Pradhan197609}, qualitative independence \cite{Cohen199411}, frameproof codes \cite{Cohen200005}, and so on.

Let $\mathcal{D}$ be a random linear code having the same average second-order weight distribution as $\rlccodeB_{n-m,n}$ with $0< m < n$. From Proposition~\ref{pr:PairwiseIndependence}, it follows that for any linearly independent $\seq{u}, \seq{v} \in \field{q}^n$, $P\{\seq{u}, \seq{v} \in \Xi_n(\mathcal{D})\} = q^{2(m-n)}$. Then the probability that $\Xi_n(\mathcal{D})$ contains a given pair of non-intersecting and linearly independent vectors is $q^{2(m-n)}$, and hence the probability that $\Xi_n(\mathcal{D})$ contains at least one pair of non-intersecting and linearly independent vectors is bounded above by
\[
(2q-1)^{n} q^{2(m-n)} = q^{-2n(1-m/n-\log_q(2q-1)/2)}.
\]
Consequently, the probability that $\Xi_n(\mathcal{D})$ is not intersecting converges to $0$ as $n \to \infty$ whenever the ratio
\[
R \eqdef \frac{m}{n} < 1 - \frac{1}{2} \log_q(2q-1). \quad \mbox{(cf. \cite[Theorem~3.2]{Cohen199411})}
\]
The right hand side of the inequality is the asymptotic (random coding) lower bound of maximum rate of linear intersecting codes. In order to finally relate this bound with the random coding rate of $\Xi_n(\mathcal{D})$, we still need to study the coding rate of sample codes in $\Xi_n(\mathcal{D})$. From Proposition~\ref{pr:PairwiseIndependence} it follows that
\[
E[|\Xi_n(\mathcal{D})|]
= \sum_{\seq{u} \in \field{q}^n} P\{\seq{u} \in \Xi_n(\mathcal{D})\}
= q^{m} + 1 - q^{m-n}
\]
and
\begin{IEEEeqnarray*}{l}
E[(|\Xi_n(\mathcal{D})| - E[|\Xi_n(\mathcal{D})|])^2] \\
\findent = E[|\Xi_n(\mathcal{D})|^2] - E[|\Xi_n(\mathcal{D})|]^2 \\
\findent = \sum_{\seq{u}, \seq{v} \in \field{q}^n} (P\{\seq{u} \in \Xi_n(\mathcal{D}), \seq{v} \in \Xi_n(\mathcal{D})\} \\
\findent \nullrel \breakop{-} P\{\seq{u} \in \Xi_n(\mathcal{D})\} P\{\seq{v} \in \Xi_n(\mathcal{D})\}) \\
\findent = \sum_{\substack{\seq{u}, \seq{v} \in \field{q}^n\setminus\{\seq{0}\}\\ \text{$\seq{u}, \seq{v}$ is linearly dependent}}} (q^{m-n} - q^{2(m-n)}) \\
\findent = q^{m-n}(q-1)(q^n-1)(1-q^{m-n}) \\
\findent < q^{m+1}.
\end{IEEEeqnarray*}
Applying Chebyshev's inequality then gives
\begin{IEEEeqnarray*}{l}
P\left\{\big||\Xi_n(\mathcal{D})|-q^m\big| \ge 2nq^{(m+1)/2}\right\} \\
\findent \le P\left\{\big||\Xi_n(\mathcal{D})|-E[|\Xi_n(\mathcal{D})|]\big| \ge nq^{(m+1)/2}\right\} \\
\findent \le \frac{1}{n^2}.
\end{IEEEeqnarray*}
Using $R=m/n$ and the simple inequality
\[
|\ln (1+x)| \le 2|x| \qquad \text{for $\textstyle |x|<\frac{1}{2}$}
\]
we finally obtain
\[
P\left\{\left|\frac{1}{n} \log_q |\Xi_n(\mathcal{D})| - R\right| < \frac{4q^{(1-nR)/2}}{\ln q}\right\} \ge 1-\frac{1}{n^2}
\]
for sufficiently large $n$. Roughly speaking, the coding rate of $\Xi_n(\mathcal{D})$ is $R$ with high probability, provided the length $n$ is large enough.

\section{Conclusion}\label{sec:Conclusion}

We established the method of second-order weight distributions. An analogue of MacWilliams identities for second-order weight distributions was proved. We computed the second-order weight distributions of several important code ensembles and discussed the application of second-order weight distribution in random coding approach. The obtained second-order weight distributions are very complex, so understanding their significance will be our future work.

\section*{Acknowledgment}

The author would like to thank Dr. T. Honold for his help in providing background information about joint weight. The author would also like to thank the Associate Editor, G. Cohen, and the anonymous reviewers for their helpful comments.

\appendices
\section{Two Lemmas for The Proof of Theorem~\ref{th:MacWilliams}}
\label{appsec:MacWilliams}

\begin{lemma}\label{le:MacWilliams1}
Let $V$ be a subspace of $\field{q}^n$. Then
\begin{equation}\label{eq:MacWilliams1a}
1_{\orth{V}}(\seq{v}') = \frac{1}{|V|} \sum_{\seq{v} \in V} \chi(\seq{v} \cdot \seq{v}')
\end{equation}
where
\begin{equation}\label{eq:MacWilliams1b}
\chi(v) \eqdef e^{2\pi i \tr(v) / p} \qquad \forall v \in \field{q}
\end{equation}
\begin{equation}\label{eq:MacWilliams1c}
\tr(v) \eqdef v + v^p + \cdots + v^{p^r-1} \qquad \forall v \in \field{q}.
\end{equation}
\end{lemma}

\begin{IEEEproof}
First note that $\tr(v)$ is an $\field{p}$-module epimorphism of $\field{q}$ onto $\field{p}$, and hence $\chi(v)$ is a homomorphism from the additive group of $\field{q}$ to $\nzcomplex$.

For a fixed $\seq{v}' \in \field{q}$, the mapping $\tau: V \to \field{q}$ given by $\seq{v} \mapsto \seq{v} \cdot \seq{v}'$ is an $\field{q}$-module homomorphism of $V$ into $\field{q}$, and hence the image set $\tau(V)$ is also a vector space over $\field{q}$, which must be either $\{0\}$ or $\field{q}$.

If $\seq{v}' \in \orth{V}$, then $\chi(\seq{v} \cdot \seq{v}') = \chi(0) = 1$ for all $\seq{x} \in V$, and hence identity \eqref{eq:MacWilliams1a} holds. If however $\seq{v}' \not\in \orth{V}$, then
\[
\frac{1}{|V|} \sum_{\seq{v} \in V} \chi(\seq{v} \cdot \seq{v}') = \frac{1}{q} \sum_{v \in \field{q}} \chi(v) = 0.
\]
The proof is complete.
\end{IEEEproof}

\begin{lemma}\label{le:MacWilliams2}
For any $\seq{u}, \seq{v} \in \field{q}^n$,
\[
\sum_{\seq{u}', \seq{v}' \in \field{q}^n} \chi(\seq{u} \cdot \seq{u}' + \seq{v} \cdot \seq{v}') \seq{x}^{\sow(\seq{u}',\seq{v}')}
= (\seq{x} \mat{K})^{\sow(\seq{u}, \seq{v})}
\]
where $\mat{K}$ is defined by \eqref{eq:MacWilliams3}.
\end{lemma}

\begin{IEEEproof}
First, we have
\begin{IEEEeqnarray*}{l}
\sum_{\seq{u}', \seq{v}' \in \field{q}^n} \chi(\seq{u} \cdot \seq{u}' + \seq{v} \cdot \seq{v}') \seq{x}^{\sow(\seq{u}',\seq{v}')} \\
\findent = \sum_{\seq{u}', \seq{v}' \in \field{q}^n} \prod_{i=1}^n \chi(u_i u_i' + v_i v_i') x_{\overline{(u_i',v_i')}} \\
\findent = \prod_{i=1}^n \sum_{u', v' \in \field{q}} \chi((u_i, v_i) \cdot (u', v')) x_{\overline{(u',v')}} \\
\findent = \prod_{i=1}^n \sum_{S \in \sop} x_S \sum_{(u', v') \in S} \chi((u_i, v_i) \cdot (u', v')).
\end{IEEEeqnarray*}
From Lemma~\ref{le:MacWilliams1}, it follows that
\begin{subnumcases}{\sum_{(u', v') \in S} \chi((u_i, v_i) \cdot (u', v')) =}
|S|, &$(u_i, v_i) \in \orth{S}$ \IEEEnonumber \\
-1, &$(u_i, v_i) \not\in \orth{S}$. \IEEEnonumber
\end{subnumcases}
Therefore we have
\begin{IEEEeqnarray*}{l}
\sum_{\seq{u}', \seq{v}' \in \field{q}^n} \chi(\seq{u} \cdot \seq{u}' + \seq{v} \cdot \seq{v}') \seq{x}^{\sow(\seq{u}',\seq{v}')} \\
\findent = \prod_{i=1}^n \sum_{S \in \sop} x_S \matentry{K}_{S, \overline{(u_i, v_i)}} \\
\findent = \prod_{T \in \sop} \left(\sum_{S \in \sop} x_S \matentry{K}_{S, T}\right)^{\pi_T(\sow(\seq{u}, \seq{v}))}
\end{IEEEeqnarray*}
as desired.
\end{IEEEproof}

\section{Properties of Uniformly Distributed Random Matrices}
\label{appsec:RLC}

\begin{theorem}\label{th:RLC2}
Let $\rmat{G}$ be a random $m\times n$ matrix uniformly distributed over the set $\field{q}^{m\times n}$ of all $m\times n$ matrices over $\field{q}$. Then for any invertible matrix $\mat{U} \in \field{q}^{m\times m}$, the product $\mat{U} \rmat{G}$ is uniformly distributed over $\field{q}^{m\times n}$.
\end{theorem}

\begin{IEEEproof}
For any matrix $\mat{U} \in \field{q}^{m\times m}$, we denote by $\mat{U}^*$ the mapping $\field{q}^{m\times n} \to \field{q}^{m\times n}$ given by $\mat{X} \mapsto \mat{U}\mat{X}$. It is clear that $\mat{U}^*$ is a surjective linear transformation for any invertible matrix $\mat{U} \in \field{q}^{m\times m}$, so that $\mat{U}\rmat{G}$ is uniformly distributed over $\field{q}^{m\times n}$.
\end{IEEEproof}

A corollary follows immediately.

\begin{corollary}[cf. \textnormal{\cite{Blinovsky200912}} and the references therein]\label{co:RLC}
Let $\rmat{G}$ be a random $m\times n$ matrix uniformly distributed over the set $\field{q}^{m\times n}$ of all $m\times n$ matrices over $\field{q}$. Then for any $\seq{x}, \seq{x}' \in \field{q}^m$ and $\seq{y}, \seq{y}' \in \field{q}^n$,
\begin{IEEEeqnarray*}{l}
P\{\seq{x}\rmat{G} = \seq{y}, \seq{x}'\rmat{G} = \seq{y}'\}\\
\findent =\left\{\begin{array}{ll}
1\{\seq{y} = \seq{0}, \seq{y}' = \seq{0}\}, &\seq{x} = \seq{x}' = \seq{0} \\
q^{-n}1\{\seq{y} = \seq{0}\}, &\mbox{$\seq{x} = \seq{0}$, $\seq{x}' \ne \seq{0}$} \\
q^{-n}1\{\seq{y}' = \seq{0}\}, &\mbox{$\seq{x} \ne \seq{0}$, $\seq{x}' = \seq{0}$} \\
q^{-n}1\{\seq{y}' = a\seq{y}\}, &\mbox{$\seq{x} \ne 0$, $\seq{x}' = a \seq{x}$ with $a \in \nzfield{q}$} \\
q^{-2n}, &\mbox{$\seq{x}, \seq{x}'$ are linearly independent.}
\end{array}\right.
\end{IEEEeqnarray*}
\end{corollary}


\begin{biographynophoto}{Shengtian Yang}
(S'05--M'06) was born in Hangzhou, Zhejiang, China, in 1976. He received the B.S. and M.S. degrees in biomedical engineering, and the Ph.D. degree in electrical engineering from Zhejiang University, Hangzhou, China in 1999, 2002, and 2005, respectively.

From June 2005 to December 2007, he was a Postdoctoral Fellow at the Department of Information Science and Electronic Engineering, Zhejiang University. From December 2007 to January 2010, he was an Associate Professor at the Department of Information Science and Electronic Engineering, Zhejiang University. Currently, he is a self-employed Independent Researcher in Hangzhou, China. His research interests include information theory, coding theory, and design and analysis of algorithms.
\end{biographynophoto}

\end{document}